\documentclass{article} 
\usepackage{times}
\usepackage{graphicx} 
\usepackage{subfigure}
\usepackage{amsmath}
\usepackage{amssymb}
\usepackage[numbers]{natbib}
\usepackage{hyperref}
\usepackage{mathtools}
\usepackage{algorithm}
\usepackage{algorithmic}
\usepackage{amsfonts}
\usepackage[russian,english]{babel}

\usepackage{epstopdf}
\usepackage[top=3cm, bottom=3cm, left=2.5cm, right=2.5cm]{geometry}

\newtheorem{theorem}{Theorem}

\newtheorem{lemma}[theorem]{Lemma}

\newtheorem{corollary}[theorem]{Corollary}

\newenvironment{proof}{\setlength\parindent{0pt}{\bf Proof.    }}{\hfill\rule{2mm}{2mm}}

\title{Large-scale Log-determinant Computation through Stochastic Chebyshev Expansions}

\author{
Insu Han \thanks{Department of Electrical Engineering, Korea Advanced Institute of Science and Technology, Korea. Emails: hawki17@kaist.ac.kr} \and 
Dmitry Malioutov \thanks{IBM T.\ J.\ Watson Research, Yorktown Heights, NY, USA, Email: dmaliout@gmail.com} \and 
Jinwoo Shin \thanks{Department of Electrical Engineering, Korea Advanced Institute of Science and Technology, Korea. Email: jinwoos@kaist.ac.kr}}

\begin{document}

\maketitle

\begin{abstract}
Logarithms of determinants of large positive definite matrices appear ubiquitously
in machine learning applications including Gaussian graphical and Gaussian process
models, partition functions of discrete graphical models, minimum-volume ellipsoids,
metric learning and kernel learning. Log-determinant computation involves the Cholesky
decomposition at the cost cubic in the number of variables, i.e., the matrix dimension,
which makes it prohibitive for large-scale applications. We propose a linear-time
randomized algorithm to approximate log-determinants for very large-scale positive
definite and general non-singular matrices using a stochastic trace approximation, called
the Hutchinson method, coupled with Chebyshev polynomial expansions that both rely on
efficient matrix-vector multiplications. We establish rigorous additive and
multiplicative approximation error bounds depending on the condition number of the
input matrix. In our experiments, the proposed algorithm can provide very high
accuracy solutions at orders of magnitude faster time than the Cholesky
decomposition and Schur completion, and enables us to compute log-determinants
of matrices involving tens of millions of variables.
\end{abstract}

\section{Introduction}
\label{submission}

Scalability of machine learning algorithms for extremely large data-sets and models has
been increasingly the focus of attention for the machine learning community, with
prominent examples such as first-order stochastic optimization methods and randomized
linear algebraic computations. One of the important tasks from linear algebra that
appears in a variety of machine learning problems is computing the log-determinant
of a large positive definite matrix. For example, serving as the normalization
constant for multivariate Gaussian models, log-determinants of covariance (and
precision) matrices play an important role in inference, model selection and learning
both the structure and the parameters for Gaussian Graphical models and Gaussian
processes \cite{rue_GMRF, rasmussen_GP, dempster_cov_sel}. Log-determinants also play
an important role in a variety of Bayesian machine learning problems, including
sampling and variational inference \cite{mackay_book}. In addition, metric and kernel
learning problems attempt to learn quadratic forms adapted to the data, and formulations
involving Bregman divergences of log-determinants have become very popular
\cite{Dhillon_metric_learning, MVE_rousseuw}. Finally, log-determinant computation
also appears in some discrete probabilistic models, e.g., tree mixture models \cite{meila2001learning,anandkumar2012learning}
and Markov random fields \cite{wainwright2006log}.
In planar Markov
random fields \cite{schraudolph2009efficient,johnson2010learning} inference and learning
involve log-determinants of general non-singular matrices.

For a positive semi-definite matrix  $B\in \mathbb R^{d\times d}$, numerical linear
algebra experts recommend to compute log-determinant using the Cholesky decomposition.
Suppose the Cholesky decomposition is $B = L L^T$, then $\log \det (B) = 2 \sum_i \log L_{ii}$.
The computational complexity of Cholesky decomposition is cubic with respect to the number of
variables, i.e., $O(d^3)$.\footnote{
For sparse matrices with a small tree-width, the complexity of Cholesky decomposition is cubic
in the tree-width.} For large-scale applications involving more than tens of thousands
of variables, this operation is not feasible. Our aim in this paper is to compute accurate
approximate log-determinants for matrices of much larger size involving {\em tens of
millions} of variables.

\vspace{0.1in}
\noindent{\bf Contribution.} Our approach to compute accurate approximations of log-determinant
for a positive definite matrix
uses a combination of stochastic trace-estimators and Chebyshev polynomial
expansions. Using the Chebyshev polynomials, we first approximate the log-determinant by the
trace of power series of the input matrix. We then use a stochastic trace-estimator,
called the {\em Hutchison method} \cite{hutchinson1989stochastic}, to estimate the
trace using multiplications between the input matrix and random vectors.
The main assumption for our method is that the matrix-vector
product can be computed efficiently. For example, the time-complexity of the proposed algorithm
grows linearly with respect to the number of non-zero entries in the input matrix. We also extend our approach
to general non-singular matrices to compute the absolute values of their log-determinants.
We establish rigorous additive and multiplicative approximation error bounds
for approximating the log-determinant under the proposed algorithm. Our theoretical results
provide an analytic understanding on our Chebyshev-Hutchison method depending
on sampling number, polynomial degree and the condition number (i.e., the ratio between
the largest and smallest singular values) of the input matrix. In particular, they
imply that if the condition number is $O(1)$, then the algorithm provides
$\varepsilon$-approximation guarantee (in multiplicative or additive) in linear time
for any constant $\varepsilon>0$.

We first apply our algorithm to obtain a randomized linear-time approximation scheme for
counting the number of spanning trees in a certain class of graphs 
where it could be used for efficient inference in tree mixture models
\cite{meila2001learning,anandkumar2012learning}. We also apply our algorithm for
finding maximum likelihood parameter estimates of Gaussian Markov random fields
of size $5000\times 5000$ (involving $25$ million variables!), which is infeasible for
the Cholesky decomposition. Our experiments show that our proposed algorithm is
orders of magnitude faster than the Cholesky decomposition and Schur completion for
sparse matrices and provides solutions with $99.9\%$ accuracy in approximation. It
can also solve problems of dimension tens of millions in a few minutes on our single
commodity computer.
Furthermore,
the proposed algorithm is very easy to parallelize and hence has a potential to handle
even a bigger size.
In particular, the
Schur method was used as a part of QUIC algorithm \cite{hsieh2013big} for sparse
inverse covariance estimation with over million variables, hence our algorithm could
be used to further improve its speed and scale.

\vspace{0.1in}
\noindent{\bf Related work.} Stochastic trace estimators have been studied in the literature
in a number of applications. \cite{bekas2007diag, malioutov2006low} have used a stochastic
trace estimator to compute the diagonal of a matrix or of matrix inverse. Polynomial
approximations to band-pass filters have been used to count the number of eigenvalues in
certain intervals \cite{eigenvalue_histograms_Saad}. 
Stochastic approximations of
score equations have been applied in \cite{stein2013stochastic} to learn large-scale Gaussian processes.
{The works closest to ours which have used stochastic trace estimators for Gaussian process parameter learning
are \cite{leithead_GP_stoch}
and \cite{aune2014_GP} which instead use Taylor expansions and Cauchy integral formula, respectively.
A recent improved analysis using Taylor expansions has also appeared  in \cite{boutsidis2015randomized}.
However, as reported in Section \ref{sec:exp}, our method using Chebyshev expansions provides much better accuracy in experiments
than that using Taylor expansions, and
\cite{aune2014_GP} need  Krylov-subspace linear system solver that is computationally expensive.
\cite{pace2004chebyshev} also use Chebyshev polynomials for log-determinant computation, but
the method is deterministic and only applicable to polynomials of small degree.
The novelty of our work is combining the Chebyshev approximation with Hutchison trace
estimators, which allows us to design a linear-time algorithm with rigorous approximation guarantees.

\vspace{0.1in}
\noindent{\bf Organization.} The structure of the paper is as follows. We introduce the necessary background
in Section \ref{sec:trace}, and describe our algorithm with approximation
guarantees in Section \ref{sec:main}. Section \ref{sec:proof} provides the proof of approximation guarantee of our algorithm, and we report experimental results in Section \ref{sec:exp}.

\section{Background}

In this section, we describe the preliminaries for our approach to
approximate the log-determinant of a \textit{positive definite} matrix.
Our approach combines the following two techniques:
(a) designing a trace-estimator for the log-determinant of positive definite matrix via Chebyshev approximation \cite{mason2002chebyshev} 
and (b) approximating the trace of positive definite matrix via Monte Carlo methods, e.g., Hutchison method \cite{hutchinson1989stochastic}.

\subsection{Chebyshev Approximation}
\label{sec:cheby}

The Chebyshev approximation technique is used to approximate analytic function with certain
orthonormal polynomials. We use $p_n(x)$ to denote the Chebyshev approximation of degree $n$
for a given function $f :[-1,1]\to \mathbb R$:
\begin{align*}
f(x) \approx p_n(x) = \sum_{j=0}^n c_j  T_j(x),
\end{align*}
where the coefficient $c_i$ and the $i$-th Chebyshev polynomial $T_i(x)$ are defined as
\begin{align}
&c_i
= \begin{dcases}
      \frac{1}{n+1}  \sum_{k=0}^n f(x_k) \ T_0(x_k) & \text{if $\ i=0$}\\
      \frac{2}{n+1}  \sum_{k=0}^n f(x_k) \ T_i(x_k) & \text{otherwise}
      \end{dcases} \\
\nonumber \\
&T_{i+1}(x) = 2 x T_i (x) - T_{i-1} (x) \label{eq:chebreq} \qquad \text{for $\ i \ge 1$}
\end{align}
where $x_k = \cos \Big( \frac{ \pi (k + 1/2 )}{ n+1} \Big)$ for $k = 0,1,2, \dots n$ and $T_0(x) = 1$, $T_1(x) = x$.

Chebyshev approximation for scalar functions can be naturally generalized to matrix functions.
Using the Chebyshev approximation $p_n(x)$ for function $f(x) = \log (1-x)$ we obtain
the following approximation to the log-determinant of a positive definite matrix $B\in \mathbb{R}^{d\times d}$:

\vspace{-0.2in}
\begin{align*}
\log \det B  &= \log \det \left( I - A \right)  = \sum_{i=1}^d \log (1-\lambda_i)\\
&\approx \sum_{i=1}^d p_n (\lambda_i)~=  \sum_{i=1}^d \sum_{j=0}^n c_j  T_j(\lambda_i)\\
&=  \sum_{j=0}^nc_j \sum_{i=1}^d   T_j(\lambda_i)
~=  \sum_{j=0}^n c_j {\tt tr}\left(T_j \left(A\right)\right),
\end{align*}

where $A=I-B$ has eigenvalues $0\leq \lambda_1 , \dots , \lambda_d\leq 1$ and
the last equality is from the fact that $\sum_{i=1}^d p(\lambda_i) = {\tt tr}(p(A))$ for
any polynomial $p$.\footnote{${\tt tr}(\cdot)$ denotes the trace of a matrix.}
We remark that other polynomial approximations, e.g., Taylor, can also be used to
approximate log-determinants. We focus on the Chebyshev approximation in this paper
due to its superior empirical performance and rigorous error analysis.

\subsection{Trace Approximation via {Monte-Carlo Method}}
\label{sec:trace}

The main challenge to compute the log-determinant of a positive definite matrix in the previous section
is calculating the trace of $T_j \left(A\right)$ efficiently without evaluating
the entire matrix $A^k$.} We consider a Monte-Carlo approach for estimating the trace of a
matrix. First, a random vector  $\mathbf{z}$ is drawn from some fixed distribution, such that
the expectation of $\mathbf{z}^\top A  \mathbf{z}$ is equal to the trace of $A$. By sampling
$m$ such i.i.d random vectors, and averaging we obtain an estimate of ${\tt tr}(A)$.

It is known that the Hutchinson method, where components of the random vectors $Z$ are i.i.d
Rademacher random variables, i.e., $\Pr(+1) = \Pr(-1) = \frac{1}{2}$,  has the smallest
variance among such Monte-Carlo methods \cite{hutchinson1989stochastic, avron2011randomized}.
It has been used extensively in many applications \cite{avron2010counting, hutchinson1989stochastic, aravkin2012robust}. Formally, the Hutchinson trace estimator ${\tt tr}_m(A)$ is known to
satisfy the following:
$$\mathbf{E} \left[ {\tt tr}_m(A) :=\frac{1}{m} \sum_{i=1}^{m} \mathbf{z}_{i}^\top A \mathbf{z}_i \right] = {\tt tr}(A)$$
$$\mathbf{Var} \left[ {\tt tr}_m(A) \right] = 2 \left( \| A \|^2 - \sum_{i=1}^n A_{ii}^2 \right).$$
{Note that computing $\mathbf{z}^\top A  \mathbf{z}$ requires only multiplications between a matrix and a vector,
which is particularly appealing when evaluating $A$ itself is expensive, e.g., $A=B^k$ for
some matrix $B$ and large $k$.
Furthermore, for the case $A=T_j \left(X\right)$, one can compute
$\mathbf{z}^\top T_j \left(X\right)  \mathbf{z}$ more efficiently
using the following recursion on the vector $w_j = T_j(X) \mathbf{z}$:
$$w_{j+1} = 2 X w_{j} - w_{j-1},$$
which follows directly from \eqref{eq:chebreq}.

\section{Log-determinant Approximation Scheme }\label{sec:main}

Now we are ready to present
algorithms to approximate the absolute value of  log-determinant of an arbitrary non-singular
square matrix $C$.
Without loss of generality, we assume that singular values of $C$ 
are in the interval $[\sigma_{\min}, \sigma_{\max}]$ for some $\sigma_{\min},\sigma_{\max}>0$, i.e.,
the condition number $\kappa(C)$ is at most $\kappa_{\max}: = \sigma_{\max}/\sigma_{\min}$.
The proposed algorithms are not sensitive to tight knowledge of $\sigma_{\min},\sigma_{\max}$, but
some loose lower and upper bounds on them, respectively, suffice.

We first present a log-determinant approximation scheme for positive definite matrices in Section \ref{sec:psd}
and that for general non-singular ones in Section \ref{sec:gen} later.

\subsection{Algorithm for Positive Definite Matrices}\label{sec:psd}
In this section, we describe our proposed algorithm for estimating the log-determinant of a
positive definite matrix whose eigenvalues are less than one, i.e., $\sigma_{\max}< 1$.
It is used as a subroutine for estimating the log-determinant of a general non-singular matrix in the next section.
The formal description of the algorithm is given in what follows.

\begin{algorithm}[tbh!]
   \caption{Log-determinant approximation for positive definite matrices with $\sigma_{\max}< 1$}
\begin{algorithmic}\label{alg1}
   \STATE {\bfseries Input:} positive definite matrix $B \in \mathbb{R}^{d \times d}$ with eigenvalues in [$\delta$ , $1-\delta$] for some $\delta> 0$,  sampling number $m$ and polynomial degree $n$
   \STATE {\bfseries Initialize:} $A  \leftarrow  I-B$, $\Gamma \leftarrow 0$
   \FOR{$i=0$ {\bfseries to} $n$}
   \STATE $c_i  \leftarrow$ $i$-th coefficient of Chebyshev approximation for $\log ( 1 - \frac{ (1-2\delta)x + 1}{2} )$
   \ENDFOR
   \FOR {$i = 1$ { \bfseries to } $m$}
   \STATE Draw a Rademacher random vector $\mathbf{v}$  and $\mathbf{u} \leftarrow c_0 \ \mathbf{v}$
   \IF{$n > 1$}
   \STATE $\mathbf{w}_0 \leftarrow \mathbf{v}$ and $\mathbf{w}_1 \leftarrow A\mathbf{v}$
   \STATE $\mathbf{u} \leftarrow \mathbf{u} + c_1 A\mathbf{v}$
   \FOR {$j = 2$ { \bfseries to } $n$}
   \STATE $\mathbf{w}_2 \leftarrow 2 A \mathbf{w}_1 - \mathbf{w}_0$
   \STATE $\mathbf{u} \leftarrow \mathbf{u} + c_j \ \mathbf{w}_2$
   \STATE $\mathbf{w}_0 \leftarrow \mathbf{w}_1$ and $\mathbf{w}_1 \leftarrow \mathbf{w}_2$
   \ENDFOR
   \ENDIF
   \STATE $\Gamma \leftarrow \Gamma + \mathbf{v^{\top}} \mathbf{u}/ m $
   \ENDFOR
   \STATE {\bfseries Output:} $\Gamma$
\end{algorithmic}
\end{algorithm}
We establish the following theoretical guarantee of the above algorithm, where its proof is given in Section \ref{sec:pfthm1}.

\begin{theorem}\label{thm:main1}
Given $\varepsilon, \zeta\in (0,1)$,
consider the following inputs for {\bf Algorithm \ref{alg1}}:
\begin{itemize}
\item $B \in \mathbb{R}^{d \times d}$ be a positive definite matrix with eigenvalues in $[\delta,1-\delta]$ for some $\delta \in (0, 1/2)$.
\item $ m \ge 54 \varepsilon^{-2} \log{\left(\frac{2}{\zeta}\right)} $
\item $ n \ge \frac{\log{\left( \frac{20}{\varepsilon} \left( \sqrt{\frac2{\delta}-1} - 1 \right) \frac{\log \left( 2 ( 1/\delta-1) \right)}{\log{\left( 1/{1-\delta}\right)}} \right)}}{\log{\left( \frac{\sqrt{2-\delta}+\sqrt{\delta}}{\sqrt{2-\delta}-\sqrt{\delta}}\right)}}
= O\left( \sqrt{\frac1{\delta}} \log \left( \frac1{\varepsilon \delta} \right) \right)$
\end{itemize}
Then, it follows that
\begin{align*}
\Pr \left[ \ \left| \log \det B - \Gamma \right| \le  \varepsilon \left| \log \det B \right| \ \right] \ge 1 - \zeta
\end{align*}
where $\Gamma$ is the output of {\bf Algorithm \ref{alg1}}.
\end{theorem}

The bound on polynomial degree $n$ in the above theorem is relatively tight, e.g., it implies to choose
$n= 14$ for $\delta=0.1$ and $\varepsilon=0.01$. However, our bound on sampling number $m$
is not, where 
we observe that $m\approx 30$ is
sufficient 
for high accuracy in our experiments.
We also remark that the time-complexity of {\bf Algorithm \ref{alg1}} is $O(mn\|B\|_0)$,
where $\|B\|_0$ is the number of non-zero entries of $B$.
This is because
the algorithm requires only multiplications of matrices and vectors.
In particular, if $m,n=O(1)$, the complexity is linear with respect to the input size.
Therefore, Theorem \ref{thm:main1} implies that one can choose
$m,n=O(1)$ for $\varepsilon$-multiplicative approximation with probability $1-\zeta$ given constants $\varepsilon,\zeta>0$.

\subsection{Algorithm for General Non-Singular Matrices}\label{sec:gen}

Now, we are ready to present our linear-time approximation scheme for the log-determinant of general non-singular
matrix $C$, through generalizing the algorithm in the previous section.
The idea is simple: run {\bf Algorithm \ref{alg1}} with normalization of positive definite matrix $C^TC$.
This is formally described in what follows.

\begin{algorithm}[tbh!]
   \caption{Log-determinant approximation for general non-singular matrices}
\begin{algorithmic}\label{alg2}
\STATE {\bfseries Input:}  matrix $C \in \mathbb{R}^{d \times d}$ with 
singular values are in the interval $[\sigma_{\min}, \sigma_{\max}]$ for some $\sigma_{\min},\sigma_{\max}>0$,
sampling number $m$ and polynomial degree $n$
\vspace{0,05in}
\STATE {\bfseries Initialize:} $B \leftarrow  \frac1{\sigma_{\min}^2+\sigma_{\max}^2}C^TC$, $\delta \leftarrow \frac{\sigma_{\min}^2}{\sigma_{\min}^2+\sigma_{\max}^2}$
\vspace{0,05in}
\STATE $\Gamma \leftarrow$ Output of {\bf Algorithm \ref{alg1}} for inputs $B, m, n, \delta$
\vspace{0,05in}
\STATE {\bfseries Output:} $\Gamma\leftarrow \left( \Gamma + d \log{(\sigma_{\min}^2+\sigma_{\max}^2)} \right) / 2$
\end{algorithmic}
\end{algorithm}
{\bf Algorithm 2} 
is motivated to design from the equality $\log|\det C| = \frac12 \log \det C^T C$.
Given non-singular matrix $C$, one need to choose appropriate $\sigma_{\max},\sigma_{\min}$ to run it. 
In most applications, $\sigma_{\max}$ is easy to choose, e.g.,
one can choose $$\sigma_{\max} = \sqrt{\|C\|_1\|C\|_{\infty}},$$
or one can run the power iteration \cite{ipsen1997computing} to estimate a better bound.
On the other hand, $\sigma_{\min}$ is relatively not easy to obtain depending on problems.
It is easy to obtain in the problem of counting spanning trees we studied in Section \ref{sec:spanning},
and it is explicitly given as a parameter in many machine learning log-determinant applications
\cite{wainwright2006log}.
In general, one can use the inverse power iteration \cite{ipsen1997computing} to estimate it. Furthermore,
the smallest singular value is easy to compute for random matrices \cite{tao2009inverse, tao2010random}
and diagonal-dominant matrices \cite{gershgorin1931uber, moravca2008bounds}. 

The time-complexity of  {\bf Algorithm \ref{alg2}} is still $O(m n \|C\|_0)$ instead of
$O(m n \|C^TC\|_0)$ since {\bf Algorithm \ref{alg1}} requires multiplication of matrix $C^TC$
and vectors.
We state the following additive error bound of the above algorithm. 

\begin{theorem}\label{thm:main2}
Given $\varepsilon, \zeta\in (0,1)$,
consider the following inputs for {\bf Algorithm \ref{alg2}}:
\begin{itemize}
\item $C \in \mathbb{R}^{d \times d}$ be a matrix such that
singular values
are in the interval $[\sigma_{\min}, \sigma_{\max}]$ for some $\sigma_{\min},\sigma_{\max}>0$.
\item $m \ge \mathcal M \left(\varepsilon, \frac{\sigma_{\max}}{\sigma_{\min}}, \zeta \right)$ and $n \ge \mathcal N \left( \varepsilon, \frac{\sigma_{\max}}{\sigma_{\min}} \right)$, where
\end{itemize}
\vspace{-0.15in}
\begin{align*}
&\mathcal M(\varepsilon, \kappa, \zeta):=\frac{14}{\varepsilon^{2}} \left( \log \left( 1 + \kappa^2\right) \right)^2 \log{ \frac2{\zeta}  }\\
&\mathcal N \left( \varepsilon, \kappa \right) := \frac{\log{\left( \frac{20}{\varepsilon} \left( \sqrt{2 \kappa^2 + 1}-1 \right) \frac{\log{( 1 + \kappa^2 )} \log(2 \kappa^2) }{\log(1+\kappa^{-2})} \right)}}{\log{\left( \frac{{\sqrt{2 \kappa^2 + 1}}+1}{{\sqrt{2 \kappa^2 + 1}}-1} \right)}}
= O \left( {\kappa} \log{ \frac{\kappa}{\varepsilon} } \right)
\end{align*}
Then, it follows that
\begin{align*}
\Pr \left[ \ \left| \log{\left( \left| \det C \right| \right)}- \Gamma \right| \le  \varepsilon d \ \right] \ge 1 - \zeta
\end{align*}
where $\Gamma$ is the output of {\bf Algorithm \ref{alg2}}.
\end{theorem}
\begin{proof}
The proof of Theorem \ref{thm:main2} is quite straightforward using Theorem \ref{thm:main1} for $B$ with the facts that
$$2\log |\det C| = \log \det B + d \log{(\sigma_{\min}^2+\sigma_{\max}^2)}$$
and $|\log \det B| \leq d \log \left(1+\frac{\sigma_{\max}^2}{\sigma_{\min}^2}\right)$.
\end{proof}

We remark that the condition number $\sigma_{\max}/\sigma_{\min}$ decides the complexity of {\bf Algorithm \ref{alg2}}. }
As one can expect, the approximation quality and algorithm complexity become worse for matrices
with very large condition numbers, as the Chebyshev approximation for the function $\log x$ near
the point $0$ is more challenging and requires higher degree approximations.

When $\sigma_{\max}\geq1$ and $\sigma_{\min}\leq 1$, i.e. we have mixed signs
for logs of the singular values, a multiplicative error bound (as stated in Theorem \ref{thm:main1})
can not be obtained since the log-determinant can be zero in the worst case.
On the other hand, when $\sigma_{\max}<1$ or $\sigma_{\min}>1$, we further show that the above
algorithm achieves an $\varepsilon$-multiplicative approximation guarantee, as stated in the
following corollaries.

\begin{corollary}\label{cor1}
Given $\varepsilon, \zeta\in (0,1)$, consider the following inputs for {\bf Algorithm \ref{alg2}}:
\begin{itemize}
\item $C \in \mathbb{R}^{d \times d}$ be a matrix such that
singular values
are in the interval $[\sigma_{\min}, \sigma_{\max}]$ for some $\sigma_{\max}<1$.
\item $m \ge \mathcal M \left( {\varepsilon\log{\frac1{\sigma_{\max}}}} , \frac{\sigma_{\max}}{\sigma_{\min}},\zeta\right)$
\item $n \ge \mathcal N \left( {\varepsilon\log{\frac1{\sigma_{\max}}}} , \frac{\sigma_{\max}}{\sigma_{\min}} \right)$
\end{itemize}
Then, it follows that
\begin{align*}
\Pr \left[ \ \left| \log{ \left| \det C \right| } - \Gamma \right| \le  \varepsilon \left| \log{ \left| \det C \right| } \right| \right] \ge 1 - \zeta
\end{align*}
where $\Gamma$ is the output of {\bf Algorithm \ref{alg2}}.
\end{corollary}

\begin{corollary}\label{cor2}
Given $\varepsilon, \zeta\in (0,1)$,
consider the following inputs for {\bf Algorithm \ref{alg2}}:
\begin{itemize}
\item $C \in \mathbb{R}^{d \times d}$ be a matrix such that
singular values
are in the interval $[\sigma_{\min}, \sigma_{\max}]$ for some $\sigma_{\min}>1$.
\item $m \ge \mathcal M \left(\varepsilon\log \sigma_{\min}, \frac{\sigma_{\max}}{\sigma_{\min}},\zeta\right)$
\item $n \ge \mathcal N \left(\varepsilon\log \sigma_{\min}, \frac{\sigma_{\max}}{\sigma_{\min}} \right)$
\end{itemize}
Then, it follows that
\begin{align*}
\Pr \left[ \ \left| \log{  \det C } - \Gamma \right| \le  \varepsilon  \log{  \det C }  \right] \ge 1 - \zeta
\end{align*}
where $\Gamma$ is the output of {\bf Algorithm \ref{alg2}}.
\end{corollary}

The proofs of the above corollaries are given in the supplementary material due to the space limitation.

\subsection{Application to Counting Spanning Trees}\label{sec:spanning}

We apply {\bf Algorithm \ref{alg2}} to a concrete problem,
where we study counting the number of spanning trees in
a simple undirected graph $G = (V , E)$ where
there exists a vertex $i^*$ such that $(i^*,j)\in E$ for all $j\in V\setminus \{i^*\}$.
Counting spanning trees is one of classical well-studied counting problems, and
also necessary in machine learning applications, e.g., tree mixture models \cite{meila2001learning,anandkumar2012learning}.
We denote
the maximum and average degrees of vertices in $V\setminus \{i^*\}$ by $\Delta_{\tt max}$
and $\Delta_{\tt avg}>1$, respectively.
In addition, we let $L(G)$ denote the Laplacian matrix of $G$. 
Then, from Kirchhoff's matrix-tree theorem, the number of spanning tree $\tau{(G)}$ is equal to
$$
\tau(G)= \det L(i^*),
$$
where
$L( i^*)$ is the $(|V|-1) \times (|V|-1)$ sub matrix of $L(G)$ that is
obtained by eliminating the row and column corresponding to $i^*$.
Now, it is easy to check that eigenvalues of $L(i^*)$ are in $[1,2\Delta_{\tt max}-1]$.
Under these observations, we derive the following corollary.

\begin{corollary}\label{cor:tree}
Given $0<\varepsilon<\frac2{\Delta_{\tt avg}-1}, \zeta \in (0,1)$,
consider the following inputs for {\bf Algorithm \ref{alg2}}:
\begin{itemize}
\item $C= L(i^*)$
\item $m \ge \mathcal M \left( \frac{\varepsilon(\Delta_{\tt avg}-1)}{4}, 2 \Delta_{\tt max} -1, \zeta\right) $
\item $n \ge \mathcal N \left(\frac{\varepsilon(\Delta_{\tt avg}-1)}{4}, 2 \Delta_{\tt max} - 1 \right)$
\end{itemize}
Then, it follows that
$$
\Pr \left[ |\log \tau(G) - \Gamma | \le \varepsilon \log \tau(G) \right] \ge 1 - \zeta
$$
where $\Gamma$ is the output of {\bf Algorithm \ref{alg2}}.
\end{corollary}
The proof of the above corollary is given in the supplementary material due to the space limitation.
We remark that the running time of {\bf Algorithm \ref{alg2}} with inputs in the above theorem
is $O(nm\Delta_{\tt avg} |V|)$. Therefore, for $\varepsilon, \zeta =\Omega(1)$ and $\Delta_{\tt avg}=O(1)$, i.e.,
$G$ is sparse, one can choose $n,m=O(1)$ so that
the running time of {\bf Algorithm \ref{alg2}} is $O(|V|)$.

\section{Proof of Theorem \ref{thm:main1}}\label{sec:proof}
In order to prove Theorem \ref{thm:main1}, we first introduce some necessary
background and lemmas on
error bounds of Chebyshev approximation and Hutchinson method we introduced in Section \ref{sec:cheby}
and Section \ref{sec:trace}, respectively.

\subsection{Convergence Rate for Chebyshev Approximation}
Intuitively, one can expect that the approximated Chebyshev polynomial converges to its original function as degree $n $ goes to $\infty$.
Formally, the following error bound is known \cite{berrut2004barycentric, xiang2010error}.

\begin{theorem}\label{thm:ckn}
Suppose $f$ is analytic with $\left| f ( z ) \right| \le M$ in the region bounded by the ellipse with foci $\pm1$ and major and minor semiaxis lengths summing to $K > 1$. Let $p_n$ denote the interpolant of $f$ of degree $n$ in th Chebyshev points as defined in section \ref{sec:cheby}, then for each $n \ge 0$,
\begin{align*}
\max_{x \in [-1,1]} \left| f(x) - p_n(x) \right| \le  \frac{4M}{\left( K-1 \right) K^{n}}
\end{align*}
\end{theorem}

To prove Theorem \ref{thm:main1} and Theorem \ref{thm:main2},
we are in particular interested in
$$f(x)=\log (1-x),\qquad\mbox{for}~x\in[\delta, 1-\delta].$$
Since Chebyshev approximation is defined in the interval $[-1,1]$, e.g., see Section \ref{sec:cheby}, 
one can use the following linear mapping $g : [\delta, 1-\delta] \rightarrow [-1,1]$ 
so that
\begin{align*}
&\max_{x \in [-1,1]} \left| ( f \circ g )(x)  - p_n(x) \right| \\
= &\max_{x \in [\delta,1-\delta]} \left| f \left( x \right) - (p_n\circ g^{-1}) (x)  \right| 
\end{align*}
For notational convenience, we use $p_n(x)$ to denote $(p_n \circ g^{-1})(x)$ in what follows. 

We choose the ellipse region, denoted by $\mathcal E_{K}$, in the complex plane with foci $\pm 1$ and its semimajor axis length is $1/(1-\delta)$ where $f \circ g$ is analytic on and inside. The length of semimajor axis of the ellipse is equal to $\sqrt{\left( 1/(1-\delta) \right)^2 - 1}$. Hence, the convergence rate $K$ can be set to
$$K = \frac1{1-\delta} + \sqrt{\left( \frac1{1-\delta}\right)^2-1}=\frac{\sqrt{2-\delta}+\sqrt{\delta}}{\sqrt{2-\delta}-\sqrt{\delta}}>1$$
The constant $M$ can be also obtained 
using the fact that 
$\left| \log z \right| = \left| \log \left| z \right| + i \arg \left( z \right) \right| \le \sqrt{\left( \log \left| z\right| \right)^2 + \pi ^2}$ for any $z \in \mathbb{C}$ as follows:
\begin{align*}
\max_{z \in\mathcal E_K} \left| (f \circ g) (z)\right| = &\max_{z \in \mathcal E_K} \left| (f \circ g) (z)\right| = \max_{z \in \mathcal E_K} \left| \log \left( 1 - g(z) \right) \right| \\
\le &\max_{z \in \mathcal E_K} \sqrt{\left( \log \left| 1 - g(z) \right| \right)^2 + \pi ^2} \\
= &\sqrt{ \log^2 \left( 2 \left( \frac1{\delta} - 1\right)\right)  + \pi^2} \le 5 \log \left( 2\left(\frac1{\delta} - 1 \right) \right) := M.
\end{align*}

Hence, for $x \in [\delta, 1-\delta]$,
$$
\left| \log \left( 1-x \right) - p_n(x) \right| \le \frac{20 \log \left( 2\left( \frac1{\delta} - 1 \right) \right)}{\left( K-1\right) K^{n}}
$$
Under these observations, we establish the following lemma that is a `matrix version' of Theorem \ref{thm:ckn}.
\begin{lemma}\label{thm:dckn} Let $B \in \mathbb{R}^{d \times d} $ be a positive definite matrix whose eigenvalues are in $[\delta, 1-\delta]$
for $\delta \in (0, 1/2)$.
Then, it holds that
$$
\left| \ \log \det B - {\tt tr} \big( p_n(I-B) \big) \right| \le \frac{20 d \log \left( 2\left(\frac1{\delta} - 1 \right) \right)}{\left( K-1\right) K^{n}}
$$
where $K = \frac{\sqrt{2-\delta} + \sqrt{\delta}}{\sqrt{2-\delta} - \sqrt{\delta}}$.
\end{lemma}

\begin{proof}
Let $ \lambda_1, \lambda_2 , \cdots , \lambda_d\in [\delta,1-\delta]$ be eigenvalues of matrix $A = I-B$.
Then, we have
\begin{align*}
\left| \log \det( I - A) - {\tt tr} \left( p_n(A) \right) \right| &= \  \left| {\tt tr} \left( \log( I - A) \right) - {\tt tr} \left( p_n(A) \right) \right| \\
&= \  \left| \sum_{i=1}^{d} \log(1- \lambda_i)  -  \sum_{i=1}^d p_n(\lambda_i) \right| \\
&\le \  \sum_{i=1}^{d} \left|  \log(1- \lambda_i)  -  p_n(\lambda_i)  \right| \\
&\le \  \sum_{i=1}^{d}  \frac{20 \log \left( 2 \left( \frac1{\delta} - 1\right) \right)}{ \left( K - 1 \right) K^{n}}
\end{align*}
where we use Theorem \ref{thm:ckn}. This completes the proof of Lemma \ref{thm:dckn}.
\end{proof}

\subsection{Approximation Error of Hutchinson Method} 
In this section, we use the same notation, e.g., $f, p_n$, used in the previous section
and we analyze the Hutchinson's trace estimator ${\tt tr}_m (\cdot)$ defined in Section \ref{sec:trace}.
To begin with, we state the following theorem that is proven in \cite{roosta2013improved}.

\begin{theorem}\label{thm:trace} Let $A \in \mathbb{R}^{d \times d}$ be a positive definite or negative definite matrix.
Given $\varepsilon_0 , \zeta_0 \in (0,1)$, it holds that
\begin{align*}
\Pr \left[ \left| {\tt tr}_m(A) - {\tt tr}(A) \right| \le \varepsilon_0 \  {\tt tr}(A) \right] \ge 1 - \zeta_0
\end{align*}
if sampling number $m$ is no smaller than $6 \ \varepsilon_0^{-2}  \log(2 / \zeta_0)$.
\end{theorem}

The theorem above provides a lower-bound on the sampling complexity of Hutchinson method,
which is independent of a given matrix $A$. To prove Theorem \ref{thm:main1}, we need an error bound
on ${\tt tr}_m (p_n(A))$. However, in general we may not know whether or not $p_n(A)$ is positive definite
or negative definite. We can guarantee that the eigenvalues of $p_n(A)$ will be negative using the following lemma.

\begin{lemma}\label{lmm:neg}
$p_n(x)$ is a negative-valued polynomial in the interval $[\delta,1-\delta]$ if
\begin{align*}
\frac{20 \log \left( 2\left(\frac1{\delta} - 1 \right) \right)}{\left( K-1\right) K^{n}} \le \log \left( \frac{1}{1-\delta} \right)
\end{align*}
where we recall that $K = \frac{\sqrt{2-\delta} + \sqrt{\delta}}{\sqrt{2-\delta} - \sqrt{\delta}}$.
\end{lemma}

\begin{proof}
From Theorem \ref{thm:ckn}, we have
\begin{align*}
\max_{ \left[ \delta, 1 - \delta \right] } p_n(x) &= \max_{ \left[ \delta, 1 - \delta \right] } f(x) + \left( p_n(x) - f(x) \right) \\
&\le \max_{ \left[ \delta, 1 - \delta \right] } f(x) + \max_{ \left[ \delta, 1 - \delta \right] } \left| p_n(x) - f(x) \right| \\
&\le \log{ \left( 1 - \delta \right) } +  \frac{20 \log \left( 2 (\frac1{\delta} - 1) \right)} { \left( K-1 \right) K^{n}} \le 0,
\end{align*}
where we use $\frac{20 \log \left( 2 (1/{\delta} - 1) \right)}{ \left( K-1 \right) K^{n}} \le -\log({1-\delta})$.
This completes the proof of Lemma \ref{lmm:neg}.
\end{proof}

\subsection{Proof of the Theorem \ref{thm:main1}}\label{sec:pfthm1}
Now we are ready to prove Theorem \ref{thm:main1}.
First, one can check that sampling number $n$ in the condition of Theorem \ref{thm:main1} satisfies
\begin{equation}
\frac{20 \log \left( 2 (\frac1{\delta} - 1) \right)} { \left( K-1 \right) K^{n}} \le \frac{\varepsilon}{2} \log \left( \frac1{1-\delta}\right).\label{eq1}
\end{equation}
Hence, from Lemma \ref{lmm:neg}, it follows that $p_n(A)$ is negative definite where $A = I - B$ and eigenvalues of $B$ are in $[\delta,1-\delta]$.
Hence, we can apply Theorem \ref{thm:trace} as
\begin{align}
\Pr \left[ \left| {\tt tr} \left( p_n(A) \right) - {\tt tr}_m \left( p_n(A) \right) \right| \le  \frac{\varepsilon}{3} \left| {\tt tr} \left( p_n(A) \right) \right| \right] \geq 1- \zeta,\label{eq2}
\end{align}
for $m \ge 54 \varepsilon^{-2}  \log{ \left( \frac2{ \zeta } \right) }$.
In addition, from Theorem \ref{thm:dckn}, we have
\begin{align*}
\left| {\tt tr} \left( p_n(A) \right) \right| - \left| \log \det B \right| 
&\le \ \left| \log \det B - {\tt tr} \left( p_n(A) \right) \right| \\
&\le \ \frac{20 d \log \left( 2 (1/\delta - 1) \right)} { \left( K-1 \right) K^{n}} \\
&\le \  \frac{\varepsilon}{2} d \log \left( \frac1{1-\delta}\right) \le \frac{\varepsilon}{2} \left|\log \det B \right|,
\end{align*}
which implies that
\begin{align}
\left| {\tt tr} \left( p_n(A) \right)\right| \le \left( \frac{\varepsilon}{2} + 1 \right) \left|\log \det B \right| \le \frac{3}{2} \left| \log \det B \right|.\label{eq3}
\end{align}
Combining \eqref{eq1}, \eqref{eq2} and \eqref{eq3} leads to the conclusion of Theorem \ref{thm:main1} as follows:
\begin{align*}
1-\zeta &\le \Pr \left[ \left| {\tt tr} \left( p_n(A) \right) - {\tt tr}_m \left( p_n(A) \right) \right| \le  \frac{\varepsilon}{3} \left| {\tt tr} \left(p_n(A) \right) \right| \right] \\
&\le \Pr \left[ \left| {\tt tr} \left( p_n(A) \right) - {\tt tr}_m \left( p_n(A) \right) \right| \le  \frac{\varepsilon}{2} \left| \log \det B \right| \right] \\
&\le \Pr [ \left| {\tt tr} \left( p_n(A) \right) - {\tt tr}_m \left( p_n(A) \right) \right| + | \log \det B - {\tt tr} \left( p_n(A) \right) | \\
&\qquad\qquad\qquad\le  \frac{\varepsilon}{2} \left| \log \det B \right| + \frac{\varepsilon}{2} \left| \log \det B \right| ] \\
&\le \Pr \left[ \left| \log \det B - {\tt tr}_m \left( p_n(A) \right) \right| \le \varepsilon \left| \log \det B \right| \right]\\
&= \Pr \left[ \left| \log \det B - \Gamma \right| \le \varepsilon \left| \log \det B \right| \right],
\end{align*}
where $\Gamma = {\tt tr}_m \left( p_n(A) \right)$.

\section{Experiments}
\label{sec:exp}

We now study our proposed algorithm on numerical
experiments with simulated and real data.

\subsection{Performance Evaluation and Comparison}
We first investigate the empirical performance of our proposed algorithm
on large sparse random matrices. We generate a random matrix $C\in \mathbb{R}^{d\times d}$,
where the
number of non-zero entries per each row is around $10$.
We first select five non-zero off-diagonal entries in each row with values
uniformly distributed in $[-1, 1]$. To make the matrix symmetric, we
set the entries in transposed positions to the same values.
 Finally, to
guarantee positive definiteness, we set its diagonal entries to absolute
row-sums and add a small weight, $10^{-3}$.

Figure \ref{fig:performance} (a) shows the running time of {\bf Algorithm \ref{alg2}}
from $d=10^3$ to $3\times 10^7$, where we choose $m=10$, $n=15$,
$\sigma_{\min}=10^{-3}$ and $\sigma_{\max}=\|C\|_{1}$. It scales roughly
linearly over a large range of sizes. We use a machine with
3.40 Ghz Intel I7 processor with $24$ GB RAM. It takes only $500$ seconds for
a matrix of size $3\times 10^7$ with $3\times 10^8$ non-zero entries. In Figure
\ref{fig:performance} (b), we study the relative accuracy compared to the exact log-determinant
computation up-to size $3 \times 10^4$. Relative errors are very small, below 0.1\%,
and appear to only improve for higher dimensions.

\begin{figure*}[t]
\begin{center}
{\includegraphics[width = \textwidth]{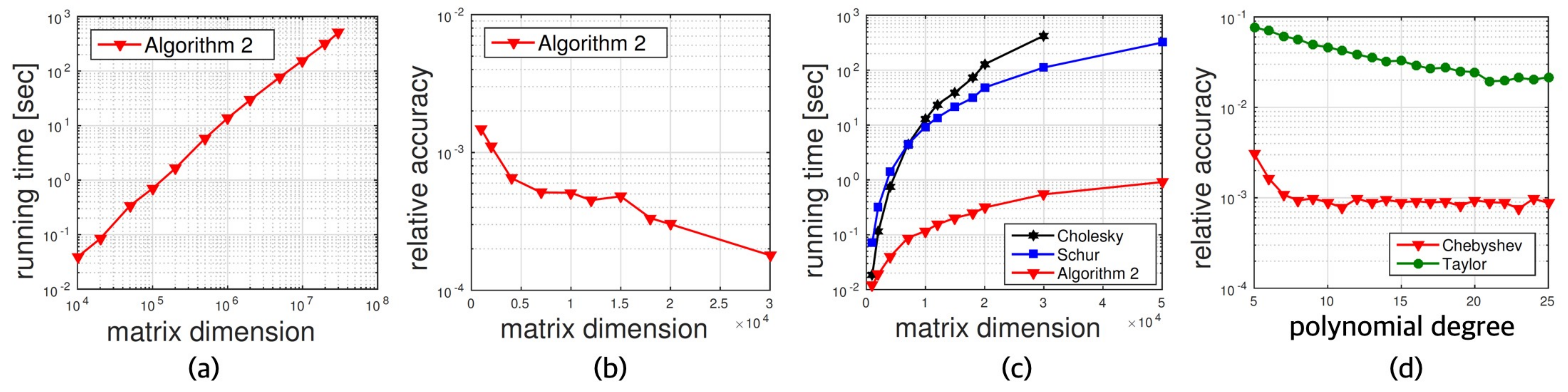}}
\vskip -0.1in
\caption{Performance evaluations of {\bf Algorithm \ref{alg2}} and comparisons with other ones:
(a) running time vs.\ dimension,  (b) relative accuracy, (c) comparison in running time with Cholesky decomposition
and Schur complement and (d) comparison in accuracy with Taylor approximation in \cite{leithead_GP_stoch}.
The relative accuracy means a ratio between the absolute error of
the output of an  approximation algorithm and the actual value of log-determinant.}
\label{fig:performance}
\vskip -0.2in
\end{center}
\vskip -0.25in
\end{figure*}

Under the same setup, we also compare the running time of our algorithm
with other algorithm for computing determinants: Cholesky decomposition and Schur
complement. The latter was used for sparse inverse covariance estimation with
over a million variables \cite{hsieh2013big} and we run the code implemented by the authors.
The running time of the algorithms are
reported in Figure \ref{fig:performance} (c). The proposed algorithm is dramatically faster than
both exact algorithms. We also compare the accuracy of our algorithm to a
related stochastic algorithm that uses Taylor expansions \cite{leithead_GP_stoch}.
For a fair comparison we use a large number of samples, $n=1000$, for both algorithms
to focus on the polynomial approximation errors. The results are reported in Figure
\ref{fig:performance} (d), showing that our algorithm using Chebyshev expansions is superior
in accuracy compared to the one based on Taylor series.

\begin{figure*}[t]
\begin{center}
\vskip -0.1in
\centerline{\includegraphics[width=\textwidth]{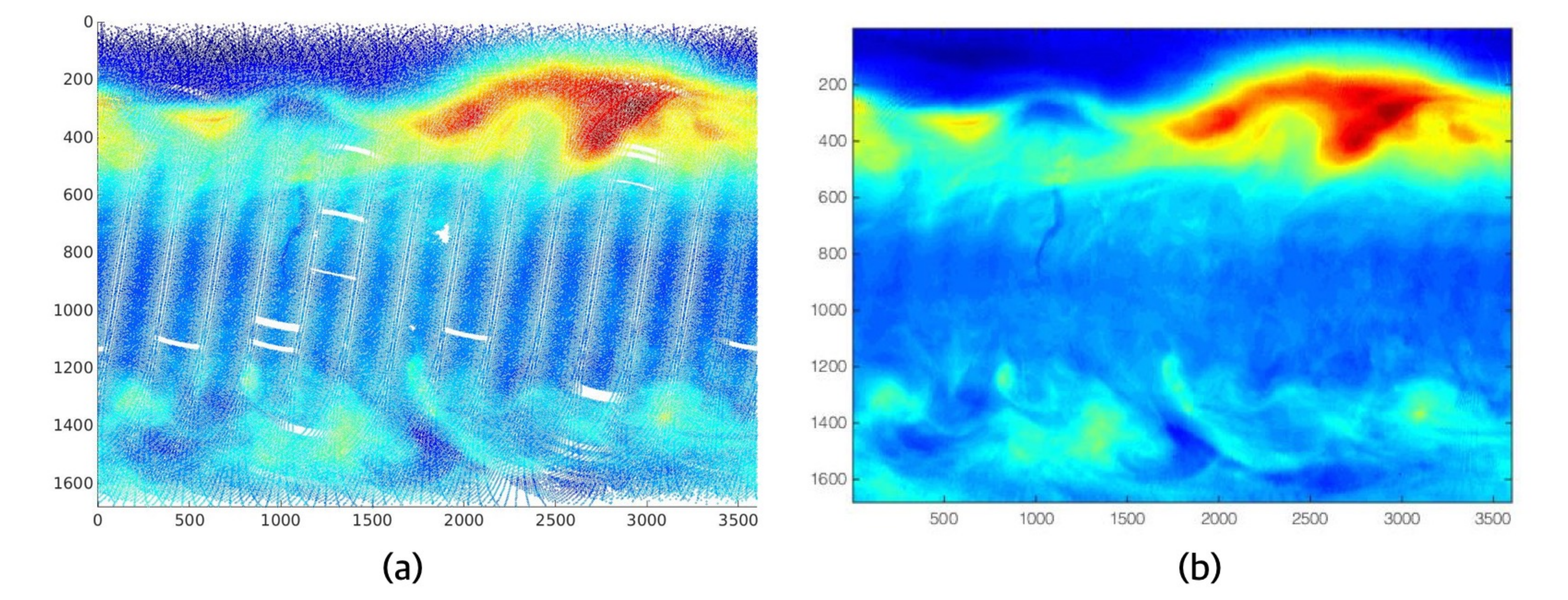}}
\vskip -0.1in
\caption{
GMRF interpolation of Ozone measurements: (a) original
sparse measurements and (b) interpolated values using a GMRF with
parameters fitted using {\bf Algorithm \ref{alg2}}.
}
\label{fig:ozone}
\end{center}
\vskip -0.35in
\end{figure*}

\subsection{Maximum Likelihood Estimation for GMRF}
\label{sec:gmrf}

\noindent{\bf GMRF with 25 million variables for synthetic data.}
We now apply our proposed algorithm for maximum likelihood (ML) estimation in Gaussian Markov Random Fields
(GMRF) \cite{rue_GMRF}. GMRF is a multi-variate joint Gaussian distribution defined with respect to a graph. Each node
of the graph corresponds to a random variable in the Gaussian distribution, where the graph captures the
conditional independence relationships (Markov properties) among the random variables. The model has been extensively used in many applications in computer vision, spatial statistics, and other fields.
The inverse covariance matrix $J$ (also called information or precision matrix) is positive definite
and sparse: $J_{ij}$ is non-zero only if the edge $\{i,j\}$ is contained in the graph.

We first consider a GMRF on a square grid of size $5000\times 5000$ (with $d= 25$ million variables)
with precision matrix $J \in \mathbb{R}^{d \times d}$ parameterized by $\rho$, i.e., each node has
four neighbors with partial correlation $\rho$. We generate a sample $\mathbf x$ from the GMRF model
(using Gibbs sampler) for parameter $\rho=-0.22$. The log-likelihood of the sample is:
$\log p({\mathbf x}| \rho) = \log \det J(\rho) - {\mathbf x}^\top J(\rho) {\mathbf x} + G$,
where $ J(\rho)$ is a matrix of dimension $25\times 10^6$ and $10^8$ non-zero entries, and $G$ is a
constant independent of $\rho$. We use {\bf Algorithm \ref{alg2}} to estimate the log-likelihood
as a function of $\rho$, as reported in Figure \ref{fig:rho}. The estimated log-likelihood is maximized at
the correct (hidden) value ${\rho} = -0.22$.

\begin{figure}[ht]
\begin{center}
\vspace{-0.05in}
\centerline{\includegraphics[width=0.65\textwidth]{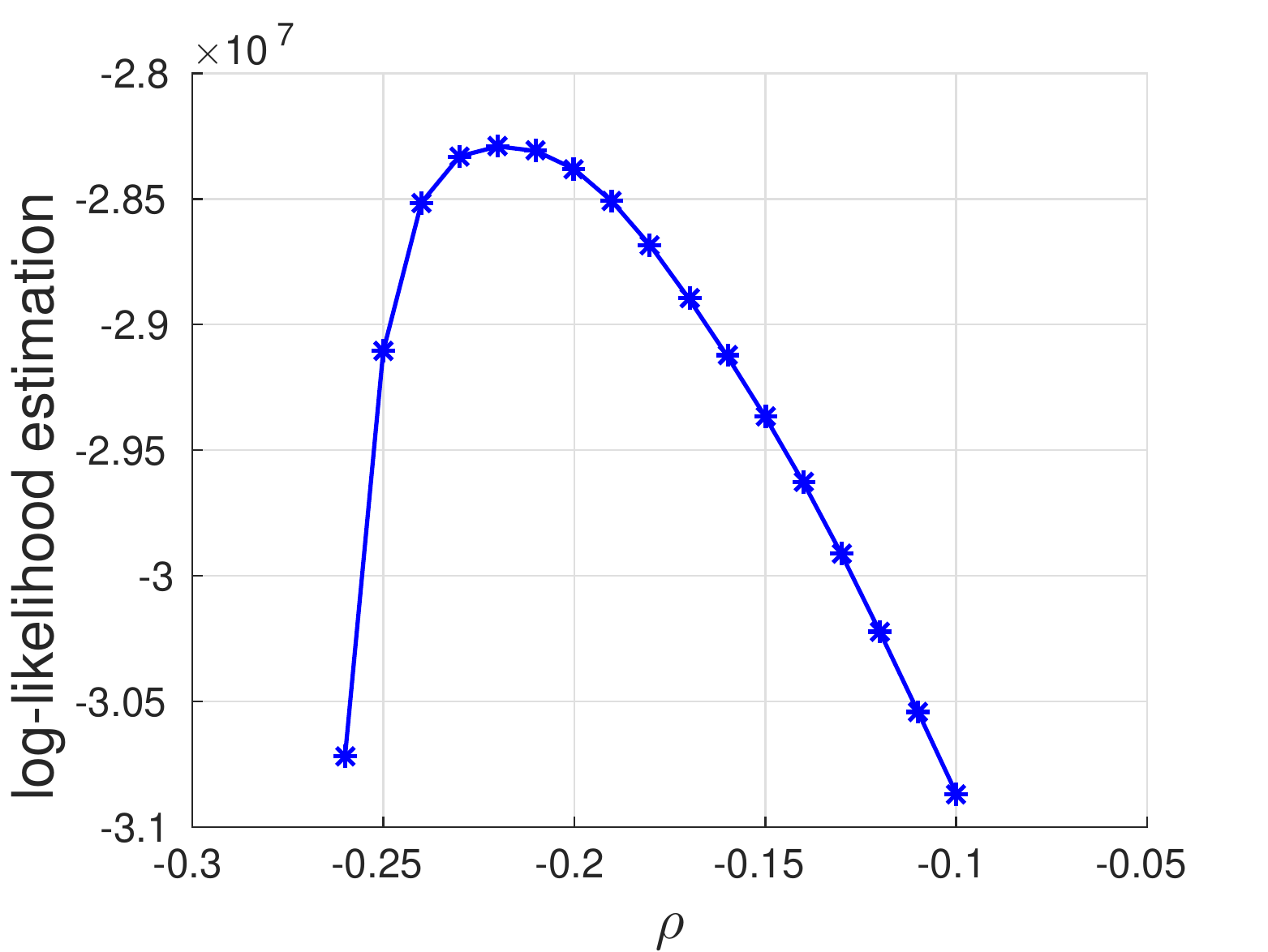}}
\vspace{-0.1in}
\caption{Log-likelihood estimation for hidden parameter $\rho$ for square GMRF model of size $5000\times 5000$.}
\label{fig:rho}
\end{center}
\vskip -0.25in
\end{figure}

\vspace{0.1in}
\noindent{\bf GMRF with 6 million variables for Ozone data.}
We also consider GMRF parameter estimation from real spatial data with missing values.
We use the data-set from \cite{aune2014_GP} that provides satellite measurements of Ozone levels
over the entire earth following the satellite tracks. We use a resolution of $0.1$ degrees in lattitude
and longitude, giving a spatial field of size $1681 \times 3601$, with over 6 million variables.
The data-set includes 172 thousands measurements. To estimate the log-likelihood in presence
of missing values, we use the Schur-complement formula for determinants. Let the precision matrix
for the entire field be $J = \left( \begin{matrix} J_o & J_{o,z}\\ J_{z,o} & J_z\end{matrix} \right)$,
where subsets $\mathbf{x}_o$ and $\mathbf{x}_z$ denote the observed and unobserved components of
$\mathbf{x}$. The marginal precision matrix of $\mathbf{x}_o$ is
$\bar{J}_{o} = J_o - J_{o, z} J_z^{-1} J_{z, o}$. Its log-determinant is computed as
$\log(\det( \bar{J}_{o})) = \log \det(J) - \log \det (J_z)$ via Schur complements. To evaluate the quadratic term
$x_o' \bar{J}_o x_o$ of the log-likelihood we need a single linear solve using an iterative solver.
We use a linear combination of the thin-plate model and the thin-membrane models
\cite{rue_GMRF}, with two parameters $\alpha$ and $\beta$:
$J = \alpha I + (\beta) J_{tp} + (1-\beta) J_{tm}$ and obtain ML estimates using {\bf Algorithm \ref{alg2}}.
Note that $\sigma_{min}(J) = \alpha$. We show the sparse measurements in Figure \ref{fig:ozone} (a) and
the GMRF interpolation using fitted values of parameters in Figure \ref{fig:ozone} (b).

\section{Conclusion}
Tools from numerical linear algebra, e.g. determinants, matrix inversion and linear solvers,
eigenvalue computation and other matrix decompositions, have been playing an important theoretical
and computational role for machine learning applications. While most matrix computations admit
polynomial-time algorithms, they are often infeasible for large-scale or high-dimensional
data-sets. In this paper, we design and analyze a high accuracy linear-time approximation
algorithm for the logarithm of matrix determinants, where its exact computation requires
cubic-time. Furthermore, it is very easy to parallelize since it requires only (separable)
matrix-vector multiplications. We believe that the proposed algorithm will find numerous
applications in machine learning problems.

\subsection*{Acknowledgement}
We would like to thank Haim Avron and Jie Chen for fruitful comments 
on Chebyshev approximations, and Cho-Jui Hsieh for providing the code for Shur 
complement-based log-det computation.

\bibliography{LogDet_arxiv}
\bibliographystyle{icml2015}

\newpage

\appendix
\section{Proof of Corollary \ref{cor1}}

For given $\varepsilon < \frac{2}{\log( \sigma_{\max}^2 )}$, set $\varepsilon_0 = \frac{\varepsilon}{2} \log \left( \frac1{\sigma_{\max}^2} \right)$. Since all eigenvalues of $C^T C$ are positive and less than 1, it follows that
\begin{align*}
&\left| \log \det \left( C^T C\right)  \right| = \left| \sum_{i=1}^d \log \lambda_i \right| \ge d \log \left( \frac1{ \sigma_{\max}^2 } \right)
\end{align*}
where $\lambda_i$ are $i$-th eigenvalues of $C^T C$. Thus, 
$$
\varepsilon_0 = \frac{\varepsilon}{2} \log \left( \frac1{\sigma_{\max}^2} \right)\le \frac{\varepsilon}{2} \frac{\left| \log \det C^T C \right|}{d} =  \varepsilon \frac{\left| \log{\left( \left| \det C \right| \right)} \right|}{d}
$$
We use $\varepsilon_0$ instead of $\varepsilon$ from Theorem \ref{thm:main2}, then following 
$$
\Pr \left[ \ \left| \log{ \left( \left| \det C \right| \right)} - \Gamma \right| \le  \varepsilon \left| \log{ \left( \left| \det C \right| \right)} \right| \ \right] \ge 1 - \zeta
$$
holds if $m$ and $n$ satifies below condition.

\section{Proof of Corollary \ref{cor2}}
Similar to proof of Corollary \ref{cor1}, set $\varepsilon_0 = \frac{\varepsilon}{2} \log \sigma_{\min}^2$. Since eigenvalues of $C^T C$ are greater than 1,
$$
\left| \log \det \left( C^T C \right) \right| \ge d \log \sigma_{\min}^2
$$
and $\varepsilon_0 \le \varepsilon \frac{\left| \log \left( | \det C | \right) \right|}{d}$. From Theorem \ref{thm:main2}, we substitute $\varepsilon_0$ into $\varepsilon$ and 
$$
\Pr \left[ \ \left| \log \det C - \Gamma \right| \le  \varepsilon \left| \log \det C \right| \ \right] \ge 1 - \zeta
$$
holds if $m$ and $n$ satifies below condition.

\section{Proof of Corollary \ref{cor:tree}}

For $\varepsilon_0=\varepsilon (\Delta_{\tt avg}-1)/2, \zeta \in (0,1)$, Theorem \ref{thm:main2} provides the following inequality:
$$
\Pr \left( |\log \det L(i^*) - \Gamma | \le \varepsilon_0 (|V| -1)\right) \ge 1 - \zeta.
$$
Observe that since vertex $i^*$ is connected all other vertices, the number of spanning tree, i.e., $\det L(i^*)$, is greater than $2^{(|V|-1)(\Delta_{\tt avg}-1)/2}$. 
Hence, we have
\begin{align*}
&\Pr \left( |\log \det L(i^*) - \Gamma | \le \varepsilon_0 (|V|-1) \right) \\
&\qquad \le \Pr \left( |\log \det L(i^*) - \Gamma | \le \varepsilon \log \det L(i^*) \right).
\end{align*}
This completes the proof of Corollary \ref{cor:tree}.

\end{document}